\documentclass[sigconf, screen=true]{acmart}
\usepackage{graphicx}
\usepackage{lipsum}
\usepackage{listings}  
\usepackage[font=small,labelfont=bf]{caption}
\newcommand{\Autoref}[1]{%
  \begingroup%
  \def\chapterautorefname{Chapter}%
  \def\sectionautorefname{Section}%
  \def\subsectionautorefname{Subsection}%
  \autoref{#1}%
  \endgroup%
}

\AtBeginDocument{%
  \providecommand\BibTeX{{%
    \normalfont B\kern-0.5em{\scshape i\kern-0.25em b}\kern-0.8em\TeX}}}

\acmConference[GLSVLSI 2022]{GLSVLSI 2022: Great Lakes Symposium on VLSI}
\acmBooktitle{GLSVLSI 2022: Great Lakes Symposium on VLSI}

\copyrightyear{2022} 
\acmYear{2022} 
\setcopyright{acmcopyright}\acmConference[GLSVLSI '22]{Proceedings of the Great Lakes Symposium on VLSI 2022}{June 6--8, 2022}{Irvine, CA, USA}
\acmBooktitle{Proceedings of the Great Lakes Symposium on VLSI 2022 (GLSVLSI '22), June 6--8, 2022, Irvine, CA, USA}
\acmPrice{15.00}
\acmDOI{10.1145/3526241.3530337}
\acmISBN{978-1-4503-9322-5/22/06}

\begin{document}

\title{Leverage the Average: Averaged Sampling in Pre-Silicon Side-Channel Leakage Assessment}

\author{Pantea Kiaei}
\affiliation{%
  \institution{Worcester Polytechnic Institute}
  \city{Worcester, MA}
  \country{USA}}
\email{pkiaei@wpi.edu}

\author{Zhenyuan Liu}
\affiliation{%
  \institution{Worcester Polytechnic Institute}
  \city{Worcester, MA}
  \country{USA}}
\email{zliu12@wpi.edu}

\author{Patrick Schaumont}
\affiliation{%
  \institution{Worcester Polytechnic Institute}
  \city{Worcester, MA}
  \country{USA}}
\email{pschaumont@wpi.edu}

\begin{abstract}
Pre-silicon side-channel leakage assessment is a useful tool to identify hardware vulnerabilities at design time, but it requires many high-resolution power traces and increases the power simulation cost of the design. By downsampling and averaging these high-resolution traces, we show that the power simulation cost can be considerably reduced  without significant loss of side-channel leakage assessment quality. We introduce a theoretical basis for our claims. Our results demonstrate up to 6.5-fold power-simulation speed improvement on a gate-level side-channel leakage assessment of a RISC-V SoC. Furthermore, we clarify the conditions under which the averaged sampling technique can be successfully used.
\end{abstract}

\vspace{-.5cm}
\vspace{-.5cm}

\vspace{-.5cm}


\maketitle

\vspace{-.3cm}
\section{Motivation}

\newif\ifalternateintro

\alternateintrotrue

\ifalternateintro

Power side-channel attacks (PSCA) are a threat to computing systems and it is imperative that vulnerabilities of hardware products to PSCA are detected as early as possible during a chip design. To study such vulnerabilities, a designer must obtain high-resolution power traces of the design for a large number of test-vectors, and apply statistical analysis on those power traces. Two common solutions to early side-channel leakage assessment are hardware prototyping and simulation. Each solution presents unique challenges as they make a different trade-off between the speed of collecting high-resolution power traces and the fidelity of their predicted power traces to those of the actual chip. 
Hardware prototyping using FPGAs is a popular technique that can be applied as soon as the RTL code is available. The power traces of the actual chip under development can be predicted by measuring the power traces of the FPGA. However, ensuring fidelity to the actual chip power traces is a challenge because of the fundamental difference between the building blocks of FPGAs (configurable logic) and those of ASIC designs (standard cell libraries). 
Power simulation of the actual gate-level netlist using CAD tools can ensure a higher fidelity to the actual chip under development. Furthermore, the simulated power traces are noiseless and thus 
represent the observation of the best-equipped attacker. A major disadvantage of simulation is that the power simulation time increases drastically at the lower abstraction levels of design.
Furthermore, simulation-based power side-channel leakage assessment can show the presence of leakage in a design, however, it cannot guarantee the absence of leakage.

In this paper, we describe a technique to decrease the power simulation time without raising the abstraction level of the simulation. Our method is to reduce the sample rate of the simulated power traces of a design, from one sample per clock cycle down to one sample over multiple clock cycles. Each of the resulting power samples thus represents the average power consumption over multiple clock cycles. We observe that reducing the sample rate in this manner significantly decreases the power simulation time. In addition, we show that the resulting traces at a reduced sample rate still return meaningful side-channel leakage assessment results. We derive the precise conditions under which this averaged sampling technique can be applied to side-channel leakage assessment. We also apply our results to several realistic case studies including a hardware crypto-coprocessor and a pipelined RISC-V processor.

Power simulation is commonly performed at three different abstraction levels: RTL, gate-level (post-synthesis/ post-layout), and transistor-level.
RTL and gate-level power simulations use event-driven simulation traces in combination with power estimators, while transistor-level power simulations are based on analog simulation. We apply our averaged sampling methodology at the gate-level, as it provides a middle ground with manageable simulation time and sufficient accuracy \cite{9184495}. For example, using Composite Current Source (CCS) delay model in synthesis has shown to generate power and timing estimates comparable to that of the analog transistor-level \cite{el2011ccs}.
In addition, the fidelity of gate-level power simulation can be increased by replacing the post-synthesis Standard Delay Format (SDF) file for a post-layout SDF file, generated after place and route of the chip under development, without any changes to the power simulation.

We study the averaged sampling technique as follows. In the next section, we review related work in pre-silicon side-channel leakage assessment. In \autoref{sec:theory}, we describe the basic concepts of averaged sampling and explain its effect on 
Correlation Power Analysis (CPA), a technique for side-channel leakage assessment. \Autoref{sec:experiments} introduces our case studies, including a pipelined RISC-V, and a hardware coprocessor design. We then conclude the paper.

\else

Power side-channel attack (PSCA) is a threat to embedded systems. While various types of PSCA have been proposed, such as differential power analysis (DPA) \cite{kocher1999differential} and correlation power analysis (CPA) \cite{brier2004correlation}, they all exploit the dependency of power consumption of a device on a secret data.
To avoid vulnerabilities of hardware products to PSCA, it is crucial to find a design's vulnerability at its earlier stages to reduce the turnaround time and the cost of production. 

Early assessment of vulnerability of a design against PSCA can be done either by prototyping using FPGAs or by power simulation using 
CAD tools. Even though prototyping using FPGAs sounds like an easily accessible approach, in practice it has several drawbacks. First, since the power measurement setup can affect the power SCA immensely, assessments on FPGA prototypes are bound to the measurement setup precision. Therefore, any less-than-perfect setup can fail to find certain vulnerable points of the design that a more potent attacker could find. 
However, assessing vulnerability to PSCA using CAD tools can estimate the power consumption of the device without being affected by the measurement precision. In other words, CAD tools can estimate the power SCA vulnerability of a device like the best-equipped attacker and maintain the highest possible signal-to-noise ratio (SNR).

Second, the process and technology of the library cells, \textit{i.e.}, building blocks, of silicon chips greatly affect the power consumption of the chips. Therefore, the fundamental distinction between the building blocks in FPGAs, \textit{i.e.}, look-up tables, and ASIC designs, \textit{i.e.}, standard cell libraries, makes the power leakage of FPGA prototypes likely to diverge from that of the ASIC implementation for the same design.
On the other hand, power simulation using CAD tools can be performed as part of the ASIC implementation using precisely the same library with which the fabrication is done.


However, performing such vulnerability assessments requires acquiring several power traces. As the time to capture a power trace is higher in simulation than in practice, an important drawback of using CAD tools for this purpose is the time it takes to generate sufficient number of traces.
Since the simulated power traces are not downgraded by measurement noise, compared to when measured, fewer traces are required to find a vulnerability. 
Power simulation can be done in three main abstraction levels: RTL-level, gate-level (post-synthesis and post-layout), and transistor-level.
Each abstraction-level is useful at different steps of the ASIC design flow. RTL and gate-level power simulations work with the same event-driven power estimators while transistor-level power traces require analog models.
Meanwhile, using Composite Current Source (CCS) delay model in synthesis has shown to generate power and timing estimates comparable to that of transistor-level \cite{el2011ccs}.
Therefore, gate-level simulation provides a middle ground solution with manageable simulation time and sufficient accuracy \cite{9184495}. 
For more accurate gate-level simulation, a post-synthesis Standard Delay Format (SDF) file can be swapped for a post-layout SDF file, generated after place and route in ASIC design flow, without any changes in the power simulation.

Furthermore, in our experiment using power simulation tools, the more number of samples per trace is requested from the tool, the longer it takes for the traces to be generated. 
Fewer samples per trace might raise questions about the accuracy of traces, however, because of the way the power samples are estimated using power simulator, we study how we can prevent losing too much accuracy and even, depending on the design, use fewer samples per trace to our advantage.
Therefore, in this work, we aim to study the effect of reducing the number of samples per simulated power trace on the ability to find power side-channel leakage (PSCL). Since power assessment at design-time is an iterative process, starting with fewer samples per trace and gradually performing more localized simulations is beneficial in reducing the design time. 

We provide a model for an ideal power simulator, show the implication of reduced number of samples in power simulation on PSCL assessment, and demonstrate our findings using commercial CAD tools. 
\todo{description of sections}

\fi  



\vspace{-.6cm}
\section{Related Work}\label{sec:related}
At the early stages of a hardware design, after behavioral simulation/verification, power side-channel leakage (PSCL) can be investigated at the RTL-level to find and harden vulnerable parts of the design. For this purpose, both RTL-PSC \cite{he2019rtl} and Param \cite{kf2020param} use the switching activity from RTL simulation to model the power consumption of a design.

{\v{S}}ija{\v{c}}i{\'c} {\em et al.} \cite{vsijavcic2020towards} introduce a simplified model of power consumption called marching stick model (MSM) for faster less accurate modeling of power side-channel. As they point out, this model is useful for RTL-level simulations and can reduce the trace generation time significantly.
However, it cannot capture the effects modeled by CAD tools in post-synthesis and later stages of the design. 
Furthermore, they demonstrate the most time-consuming part of the PSCL at design time is the power simulation using CAD tools. 

The use of CAD tools for gate-level simulation is adopted in ACA \cite{yao2020architecture} and Karna \cite{slpsk2019karna}.
ACA finds the leaky gates in the gate-level netlist and replaces them locally with hardened structures. Karna finds the leaky parts of the implementation and replaces each gate with their lower-power versions if available in the standard cell library.
Recently, the authors of Saidoyoki\cite{kiaeisaidoyoki} showed simulated power traces can provide sufficient information for successful attacks with very few downsampled traces which greatly reduces the power simulation time as the bottleneck of using power simulation CAD tools in PSCL assessment at design-time.

Most pre-silicon PSCL assessment tools have explored the different abstraction layers (RTL, gate-level, transistor-level) \cite{buhan2021sok}.
To the best of our knowledge, our work is the first paper to study power simulation cost from the dimension of time resolution and its implication on pre-silicon PSCL assessment.
As increasingly more researchers are emphasizing the integration of PSCL evaluation into ASIC design flow, it is crucial to study the implication of simulated power traces using CAD tools on leakage analysis.

\vspace{-.5cm}
\section{Theoretical Background}
\label{sec:theory}



\subsection{Power Side-Channel Analysis}

\textbf{Problem Statement.}
PSCAs find sensitive information processed in a device by a divide and conquer approach. Even cryptographic modules with long secret keys are susceptible to such attacks. In AES-128 (128-bit secret key) for instance, a PSCA can focus on finding one byte of the key at a time and therefore reduce the search space from $2^{128}$ to $16\times 2^{8}$. 
In our evaluations, we focus on CPA as the most prominent form of PSCA. However, the same conclusions can be made for other statistical moment-based measures of leakage such as DPA \cite{kocher1999differential} and TVLA \cite{Becker2013TestVL}. 
We specifically avoid TVLA in our evaluations to prevent the known false-positive issues related to this test \cite{standaert2018not} to affect our conclusions.



\textbf{CPA.}
Correlation power analysis (CPA) \cite{brier2004correlation} finds all the subkeys processed on a device by assuming a leakage model for the device. Through calculating the Pearson correlation coefficient ($\rho$) between the collected power dissipation from the device for random known inputs and the constructed leakage for each subkey guess, an attacker finds the correct subkey as the one resulting in the highest absolute $\rho$ value. 

\vspace{-.3cm}
\subsection{Simulating Power Traces}

Side-channel leakage assessment
requires accurate modeling of the statistical properties of data-dependent power effects, and requires power consumption analysis over many different test-vectors. The data-dependent power effects that may count as side-channel leakage include data-dependent logic transitions, glitches, static leakage, propagation delays, and parasitic coupling. 
We use gate-level power modeling, as it is able to capture most of the data-dependent logic effects, while at the same time being much more efficient than transistor-level (SPICE-level) simulation. Time-based gate-level power modeling is supported in commercial tools such as Cadence Joules and Synopsys Primepower.

A brief review of time-based gate-level power estimation follows. A simulated power side-channel analysis campaign on a circuit collects $N$ power traces of $K$ frames each, for a total of $N.K$ power estimations. Each power estimation determines the average power consumption of the circuit within that frame. At gate level, the circuit power consumption includes three components: switching power, internal power, and leakage power (\autoref{fig:powersources}). These factors are determined over each gate in the design, and the per-gate contributions are added up to yield circuit power. The switching power of a gate depends on the per-frame toggle rate of the output pin(s) and the capacitive load of the output pin(s). The internal gate power depends on the per-frame, per-pin input toggle rate. The gate leakage power depends on the per-frame state-dependent leakage power. The power simulator uses a technology library to reflect proper scaling factors for each gate configuration and the circuit operating conditions (temperature, voltage).

\begin{figure}
    \centering
    \vspace{-.2cm}
    \includegraphics[width=.75\linewidth]{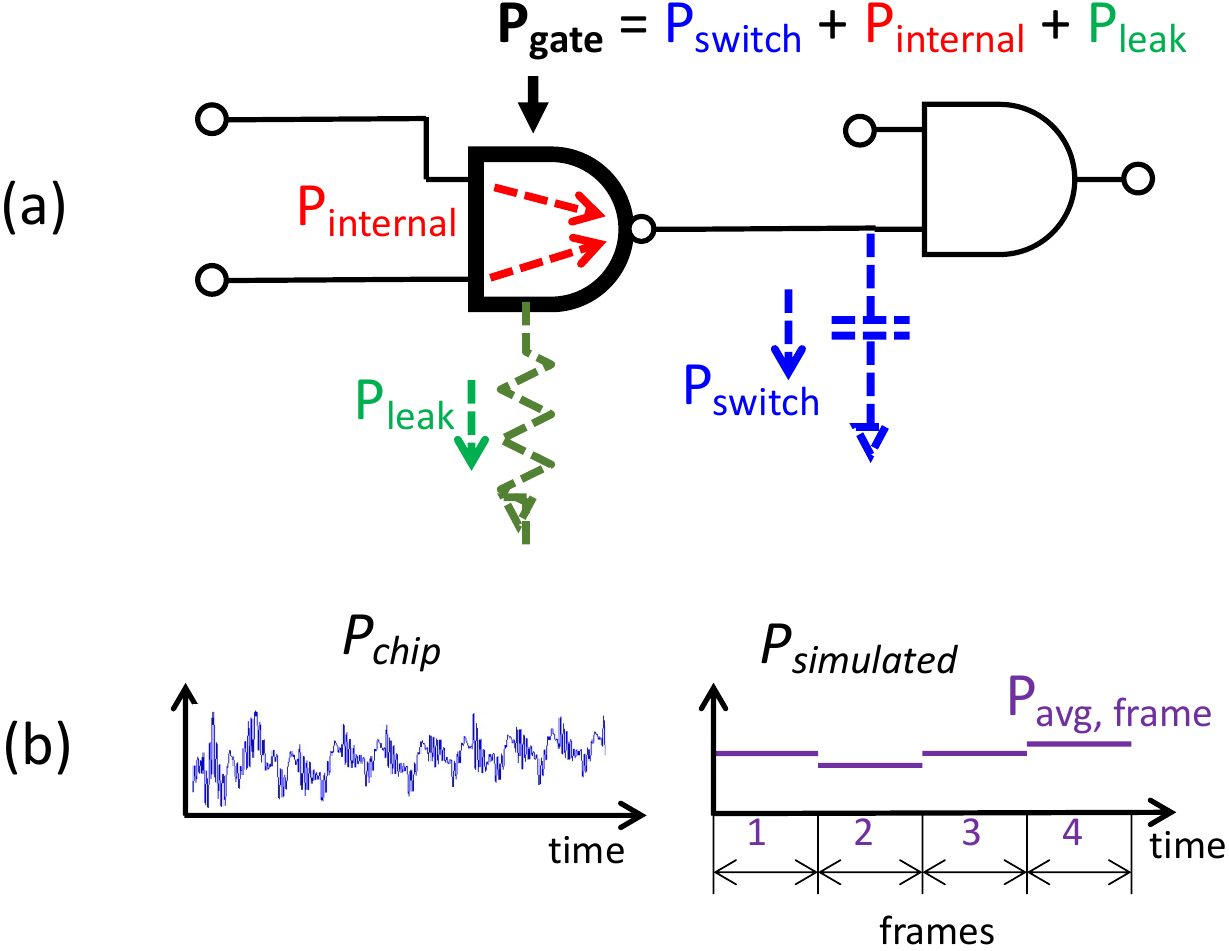}
    \vspace{-.4cm}
    \caption{(a) Gate-level power estimation for side-channel leakage captures per-gate and per-event switching power, leakage-power and internal power. (b) Gate-level power estimation partitions time in frames and determines average circuit power per frame to construct a power trace.}
    \label{fig:powersources}
    \vspace{-.2cm}
\end{figure}

\begin{table}[!t]\centering
\vspace{-.3cm}
\caption{Normalized complexity of power estimation time for three different design sizes and four different frame widths.}\label{tab:framesize}
\vspace{-.4cm}
\scriptsize
\begin{tabular}{lrrr}\toprule
\textbf{Frames} & \textbf{1 Counter} & \textbf{8 Counters} & \textbf{32 Counters} \\\midrule
\textbf{1}    & 1.0    &  7.5     &   22.7  \\
\textbf{10}   & 1.3    &  8.5     &   32.0  \\
\textbf{100}  & 3.6    &  29.7    &   139.3 \\
\textbf{1000} & 20.8   &  202.6   &   969.1 \\
\bottomrule
\end{tabular}
\vspace{-.6cm}
\end{table}

The power traces in a simulated power side-channel analysis campaign are obtained in two steps. First, a logic-level simulation records the gate and net activities for each of the $N$ test vectors in a Value Change Dump (VCD) format. Next, the VCD stimuli are analyzed for the gate-level netlist to obtain per-frame power traces. 

We observe that a larger frame window size may shorten the power estimation time, and illustrate this effect in \autoref{tab:framesize}. Three different designs containing one, eight or thirty-two 32-bit counters are implemented in SkyWater 130nm standard cells. Next, their power consumption over a 10-million clock cycle testbench is estimated by Cadence Joules (v20.01). We computed traces containing 1, 10, 100, and 1000 frames. \autoref{tab:framesize} shows a strong dependence of power estimation time over design size, which is expected: larger designs contain more events, and therefore take more time for power estimation. However, \autoref{tab:framesize} also shows a strong dependence over the number of frames. Indeed, all events corresponding to a single gate or net within the same frame are accumulated into the toggle rate for that frame. Reducing the number of frames over a set of events results in reducing the number of times a power model will be computed. This observation encourages us to investigate the impact of downsampling on the quality of side-channel leakage assessment. Clearly, power traces with fewer samples provide a simulation time advantage; the question is if they are still useful to identify side-channel leakage. 



\vspace{-.2cm}
\subsection{Sampling Power Traces}
In this section, we discuss non-averaged and averaged sampling. Non-averaged sampling is prevalent in power measurement using oscilloscopes where each power sample is a snapshot of the power consumption in a moment of time. Averaged sampling is how power simulation tools produce traces where each sample is the averaged power consumption of the design during the time interval of a frame. In the following, we give a mathematical model for each sampling type and study its implication on CPA as a measure of leakage assessment.
A summary of the analyses is given in \autoref{tab:summary_models}.
Note that we do not perform averaging as an extra step in our proposed method. Rather, we start from the averaged traces as the output of commercial power simulators.

\begin{table}[!htp]\centering
\vspace{-.2cm}
\caption{Summary of sample types}\label{tab:summary_models}
\vspace{-.4cm}
\scriptsize
\begin{tabular}{lrrr}\toprule
Sampling Type & Power Model & Implications for CPA \\\midrule
1) Non-averaged & \autoref{eq:smpl} & \autoref{th:cpa_nonave} \\
2) Averaged - Constructive & \autoref{eq:avg_same_key_rand} & \autoref{th:cpa_constr} \\
3) Averaged - Destructive & \autoref{eq:avg_multikey} & \autoref{th:cpa_destr} \\
\bottomrule
\end{tabular}
\vspace{-.2cm}
\end{table}

\vspace{-0.1cm}
\subsubsection{Non-Averaged Sampling}
In non-averaged sampling, we use the following model for each power sample:
\begin{equation}
\small\label{eq:smpl}
    P_{n,{t_s}} = \alpha . \zeta(x_n, k_{c_{t_s}}) + r_{n,{t_s}} + \mu_r
\end{equation} \normalsize
where $P_{n,{t_s}}$ denotes the power consumption at time sample $t_s$ of the $n^{th}$ power trace, $\alpha$ is an unknown constant, $\zeta(x_n, k_{c_{t_s}})$ is the selection function dependent on the controllable variable in trace n ($x_n$) and the correct key in time sample $t_s$. We model the algorithmic noise as a zero-mean random noise $r_{n,{t_s}}$ and a constant unknown bias $\mu_r$.

This model is similar to the model used by Fei {\em et al.} \cite{fei2014statistics}. 
However, each trace may contain leakage of multiple key bytes, so that one trace can correlate with different correct key values. An example is the SubBytes function in the AES-128 algorithm, which will process 16 key bytes in every round.

\textbf{Implications for CPA attack.} 
Next, we derive the performance of a CPA attack on non-averaged traces.
\begin{theorem}\label{th:cpa_nonave}
In non-averaged sampling, a sufficiently large number of traces is required for CPA attack to be successful.
\end{theorem}
\begin{proof}
The correlation coefficient between the selection function with key guess and the power trace is the covariance between the two sets divided by their standard deviations:
\small\begin{equation} \label{eq:rho}
    \begin{split}
        \rho = \frac{\sum_{j=1}^{N} \left( \left( \Delta\zeta_{j,k_g} \right) \left( P_{j,t_s} - \mathbb{E}_n[{P_{n,t_s}}] \right) \right) }{\sqrt{\sum_{j=1}^{N}\left( \Delta\zeta_{j,k_g} \right)^2 \sum_{j=1}^{N}\left( P_{j,t_s} - \mathbb{E}_n[{P_{n,t_s}}] \right)^2}}
    \end{split}
\end{equation}\normalsize
where $k_g$ is the key guess, and  
$\Delta\zeta_{j,k_g} = \zeta(x_j,k_g) - \mathbb{E}_n[{\zeta(x_n,k_g)}]$.
The only term affected by the random noise in $\rho$ is
\(P_{j,t_s}-\mathbb{E}_n[{P_{n,t_s}}] = \alpha.\Delta\zeta_{j,k_{c_{t_s}}} + r_{j,t_s} - \frac{1}{N}\sum_{n=1}^{N}r_{n,t_s}\)
where
$\Delta\zeta_{j,k_{c_{t_s}}} = \zeta(x_j, k_{c_{t_s}}) - \frac{1}{N} \sum_{n=1}^{N}\zeta(x_n,k_{c_{t_s}}).$
Because of the dependence of $P_{j,t_s}$ on random noise ($r_{j,t_s}$), the covariance term is imprecise and requires sufficiently large number of traces to converge. 
\end{proof}

\subsubsection{Averaged Sampling}
In averaged sampling, the power sample for time interval ($t_{s_1},t_{s_2}$) in the $n^{th}$ power trace is:
\small\begin{equation*} \label{eq:avg_smpl}
    \begin{split}
        & P_{n,{t_{s_1},t_{s_2}}} = \mathbb{E}_{t_s \in (t_{s_1},t_{s_2})}[P_{n,{t_s}}] \\ 
        & = \alpha. \mathbb{E}_{t_s \in (t_{s_1},t_{s_2})}[{\zeta(x_n, k_{c_{t_s}})}]
        + \mathbb{E}_{t_s \in (t_{s_1},t_{s_2})}[{r_{n,{t_s}}}] + \mu_r .
    \end{split}
\end{equation*}\normalsize

We separately analyze two scenarios in the following: 
    1)~only one key value is processed during the time interval,
    2)~multiple key values are processed during the time interval.

\subsubsection*{Case 1) One key value during the time interval}
If a single key value ($k_{B_1}$) is processed during the interval ($t_{s_1}$, $t_{s_2}$), the power sample in this interval will be:
\small\begin{equation} \label{eq:avg_same_key_rand}
    P_{n,{t_{s_1},t_{s_2}}} = \alpha. \zeta(x_n, k_{B_1}) + \mathbb{E}_{t_s \in (t_{s_1},t_{s_2})}[{r_{n,{t_s}}}] + \mu_r .
\end{equation}\normalsize

\begin{theorem}
With enough time samples processing the same key value, the averaged power samples are shifted by the constant bias in the algorithmic noise.
\end{theorem}
\begin{proof}
For a sufficiently long interval, we have:
\small\begin{equation*}
    \lim_{\Delta_{t_s} \to \infty} \mathbb{E}_{t_s \in (t_{s_1},t_{s_1}+\Delta_{t_s})}[{r_{n,{t_s}}}] = 0.
\end{equation*}\normalsize

In practice, given a long enough time interval correlating to one key value, \textit{i.e.}, $\Delta_{B_1} = t_{s_2} - t_{s_1} > \Delta_{thr}$, we have:
\small\begin{equation*} \label{eq:avg_same_key}
    P_{n,{t_{B_1},t_{B_1}+\Delta_{B_1}}} \approx \alpha. \zeta(x_n, k_{B_1}) + \mu_r .
\end{equation*}\normalsize
\end{proof}

\begin{definition}[constructive samples]
Since power samples corresponding to the same key value make the correlation between the power trace and key value stronger by removing the random noise, we name them \textbf{\textit{constructive samples}}.

\end{definition}

\textbf{Implications for CPA attack.}
Assuming averaged sampling of constructive samples, the following theorem holds.
\begin{theorem} \label{th:cpa_constr}
Given enough constructive samples in a power trace by design, 
averaged sampling will require fewer traces for CPA to succeed than non-averaged sampling.
\end{theorem}
\begin{proof}
The correlation coefficient between the selection function with key guess and the power trace is the same as \autoref{eq:rho} replacing the power samples with \autoref{eq:avg_same_key_rand}.
Given enough constructive samples to be averaged in the interval ($t_{s_1},t_{s_2}$), we have:
\small\[ P_{j,t_{s_1},t_{s_2}} - \mathbb{E}_n[{P_{n,t_{s_1},t_{s_2}}}] = \alpha.\left(\zeta(x_j,k_{B_1}) - \mathbb{E}_n[{\zeta(x_n,k_{B_1})}] \right).\]\normalsize
Therefore, the correlation is not affected by the algorithmic noise and CPA can converge with fewer traces. 

However, without enough constructive samples, we have:
\small\begin{equation*}
    \begin{split}
        P_{j,t_{s_1},t_{s_2}} - \mathbb{E}_n[{P_{n,t_{s_1},t_{s_2}}}] & = \alpha.\left(\zeta(x_j,k_{B_1}) - \mathbb{E}_n[{\zeta(x_n,k_{B_1})}] \right) \\ 
        & + \mathbb{E}_{t_s \in (t_{s_1},t_{s_2})}{ [r_{j,t_s}-\mathbb{E}_n[{r_{n,t_s}}]]}.
    \end{split}
\end{equation*}\normalsize
Therefore a sufficiently large number of traces is required for the random noise to reach its zero mean and render CPA successful.
\end{proof}

\subsubsection*{Case 2) Multiple key values during the time interval}
We assume key value $k_{B_1}$ is processed in the first part of the interval ($t_{s_1}$, $t_{s_m}$) and other key values from the set ${\{k_{B_1}\}}'$  (${k_{B_1} \not\in {\{k_{B_1}\}}'}$) are processed in the rest of the interval ($t_{s_m}$, $t_{s_2}$). 
To further simplify the analysis, we assume the time samples in interval ($t_{s_m}$, $t_{s_2}$) correspond to processing the same key value $k_{{B_1}}' \neq k_{B_1}$. The power sample in interval ($t_{s_1}$, $t_{s_2}$) can be written as:
\small\begin{equation}\label{eq:avg_multikey}
    \begin{split}
        P_{n,{t_{s_1},t_{s_2}}} & = \alpha. \zeta(x_n, k_{B_1}) + \alpha. \zeta(x_n, k_{{B_1}}') \\
        & + \mathbb{E}_{t_s \in (t_{s_1},t_{s_2})}[{ r_{n,t_s} }] + \mu_r .
    \end{split}
\end{equation}\normalsize

\begin{definition}[destructive samples]
Since the power samples corresponding to different key values being processed makes the correlation between the power trace and the secret key value weaker, we name them \textbf{\textit{destructive samples}}.
\end{definition}

\vspace{-.2cm}
\textbf{Implications for CPA attack.}
Similarly, \autoref{th:cpa_destr} is derived for CPA attack.
\begin{theorem} \label{th:cpa_destr}
Averaging destructive samples diminishes the success probability of CPA attack even with a large number of power traces.
\end{theorem}
\begin{proof}
The correlation coefficient between the selection function and the power trace is the same as the previous case.
Similarly, following the simplified averaged sample in \autoref{eq:avg_multikey}, we have:
\small\begin{equation*}
    \begin{split}
        P_{j,t_{s_1},t_{s_2}} - \mathbb{E}_n[{P_{n,t_{s_1},t_{s_2}}}] & = \alpha .\Delta\zeta_{j, k_{B_1}} + \alpha .\Delta\zeta_{j, k_{{B_1}}'}  \\
        & + \mathbb{E}_{t_s \in (t_{s_1},t_{s_2})}[{ r_{j,t_s}-\mathbb{E}_n[{r_{n,t_s}}] }]
    \end{split}
\end{equation*}\normalsize
where $\Delta\zeta_{j,k_{B}} = \zeta(x_j,k_{B}) - \mathbb{E}_n[{\zeta(x_n,k_{B})}]$.
Therefore, even with enough traces to cancel out the random noise, the term $\alpha .\Delta\zeta_{j, k_{{B_1}}'}$ greatly impairs the correlation.
\end{proof}
\vspace{-.5cm}
\subsection{Empirical verification of theorems}
\textbf{Setup.}
Following the assumption that at each clock trigger there is a rush of current in the circuit, we generate power traces in the shape of sawtooth function where the amplitude is a function of the current level and the period is similar to the system clock period. This model follows the so-called Delta-I noise \cite{gonzalez1999low} and is similar to the model used by Tiran et al. \cite{tiran2014model}.
SCA attacks, like DPA and CPA, exploit the dependency of the power consumption of a circuit on the secret data as well as a controlled data. To make our generated power traces applicable to such attacks, we further make the amplitude of the samples in each trace dependent on a function of a secret and a controlled byte. For our traces, we choose $f_{n,t_s} = xor(x_n, k_{t_s})$ where $x_n$ is the controlled value for the $n^{th}$ trace and $k_{t_s}$ is the secret value for the time sample $t_s$.
Each trace corresponds to one controlled value and contains 64 clock cycles. This is similar to measured traces in practice where an attacker feeds a device with known inputs and acquires one trace for each given input. Every 16 consecutive clock cycles correspond to one key byte value (constructive). 
These traces are interpolated in Matlab up to a continuous-time power trace such that we can create both a non-averaged sample trace (oscilloscope) as well as an averaged sample trace (CAD power simulation) from the same source version.  
We generate 256 such traces, one for each controlled byte value. 


\textbf{Experiment 1 - averaged vs. non-averaged sampling.}
For each of the 256 traces, we generate noisy continuous-time traces with Signal-to-Algorithmic Noise Ratio (SANR)\footnote{We differentiate between the SNR pertaining to the measurement noise and the SNR pertaining to the algorithmic noise, calling the latter Signal-to-Algorithmic Noise Ratio (SANR). Algorithmic noise is any signal not correlated with the target of the attack.} values of \{10, 5, 1, 0.5, 0.25, 0.1\}.
For each SANR value we create two power traces from the same source data: a non-averaged sampled version with two samples per clock cycle, and an averaged sampled version with two frames per clock cycle. 
We run CPA on the two sets of noisy traces for each key byte separately and calculate the rank of the correct key value. We repeat the same process of noisy trace generation, sampling, and CPA 100 times and report the average correct key rank and average success rate (SR) calculated as $\text{SR} = \frac{\text{number of successful attacks}}{\text{total number of attacks}}$ 
in \autoref{fig:ave_vs_nonave}. 

\begin{figure}[t]
\vspace{-.5cm}
    \centering
    \includegraphics[width=\linewidth]{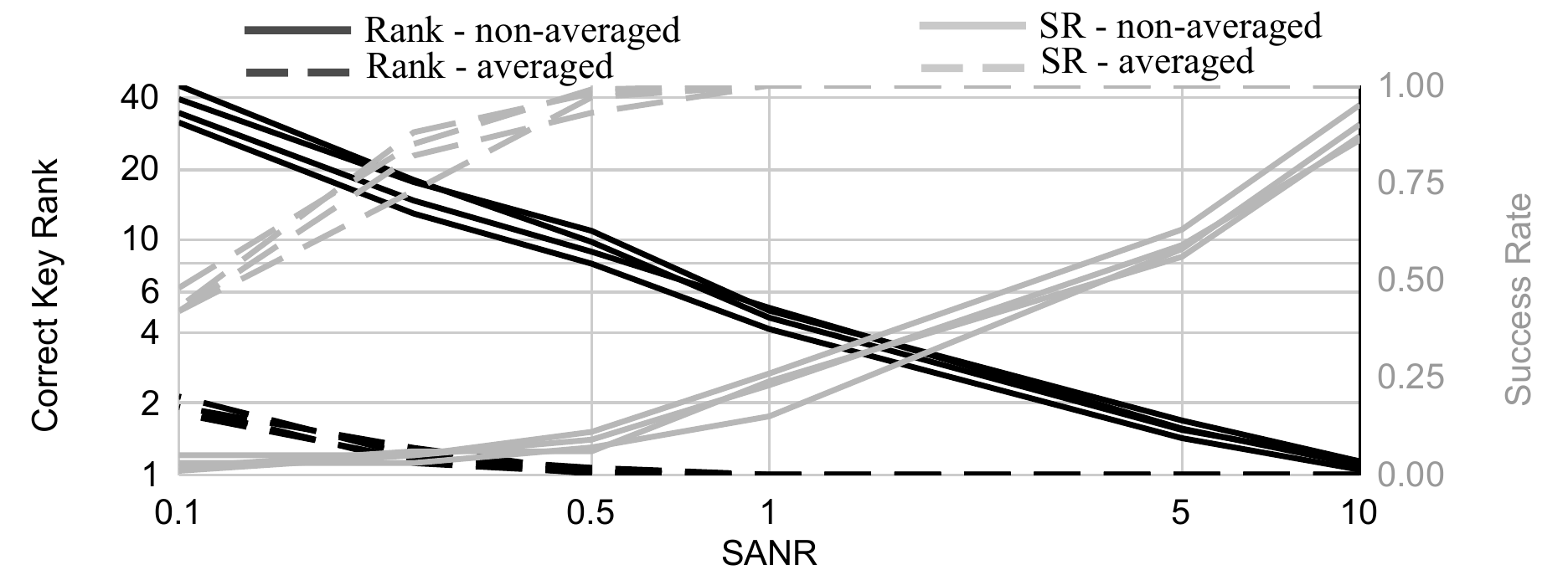}
    \vspace{-.7cm}
    \caption{Correct key rank and SR of CPA attack on 256 noisy traces with different SANR values for the four different key bytes. Solid (resp. dashed) lines show results for non-averaged (resp. averaged) sampled traces. Averaged over 100 runs.}
    \label{fig:ave_vs_nonave}
    \vspace{-.5cm}
\end{figure}

This experiment shows a substantial improvement of averaged traces over non-averaged traces. For example, attacks on averaged traces approach a success rate of 1 as soon as SANR$>0.5$, while non-averaged traces achieve the same success rate only at SANR$>10$. 


\textbf{Experiment 2 - constructive samples.}
For the same set of power traces with SANR=0.1 (weakest SANR), we average up to 16 clock cycles corresponding to the same key byte (constructive samples). 
We report the SR of CPA attack to compare the different levels of constructive averaging in \autoref{tab:SR_constructive}. As this table shows, as we average more constructive samples, it becomes more likely to find the correct key. 

\begin{table}[h]\centering
\vspace{-.3cm}
\caption{SR of CPA on key bytes in simulated traces for different number of averaged constructive samples with SANR$=0.1$ (averaged over 100 runs).}\label{tab:SR_constructive}
\vspace{-.4cm}
\scriptsize
\begin{tabular}{lrrrrr}\toprule
\textbf{\# Constructive} &\textbf{SR byte 1} &\textbf{SR byte 2} &\textbf{SR byte 3} &\textbf{SR byte 4} \\\midrule
2 &0.68 &0.61 &0.68 &0.66 \\
4 &0.78 &0.78 &0.79 &0.71 \\
8 &0.83 &0.92 &0.83 &0.81 \\
16 &0.92 &0.96 &0.9 &0.95 \\
\bottomrule
\end{tabular}
\vspace{-.3cm}
\end{table}


\textbf{Experiment 3 - destructive samples.}
We average 16 constructive cycles for each key byte and add up to 16 destructive cycles to the averaged sample to study the effect of destructive samples on CPA results. 
\autoref{tab:SR_destructive} shows the SR decreases as more destructive samples are added, attesting to \autoref{th:cpa_destr}. 

\begin{table}[h]\centering
\vspace{-.3cm}
\caption{SR of CPA on key bytes in simulated traces for different ratio of destructive to constructive samples with SANR=0.1 (averaged over 100 runs).}\label{tab:SR_destructive}
\vspace{-.4cm}
\scriptsize
\begin{tabular}{lrrrrr}\toprule
\textbf{$\frac{\text{\# Destructive}}{\text{\# Constructive}}$} &\textbf{SR byte 1} &\textbf{SR byte 2} &\textbf{SR byte 3} &\textbf{SR byte 4} \\\midrule
0 &0.93 &0.93 &0.91 &0.9 \\
0.25 &0.87 &0.79 &0.86 &0.86 \\
0.5 &0.58 &0.6 &0.68 &0.68 \\
0.75 &0.39 &0.36 &0.62 &0.57 \\
1 &0.04 &0.03 &0.06 &0.07 \\
\bottomrule
\end{tabular}
\vspace{-.5cm}
\end{table}


\vspace{-.3cm}
\section{Case Studies} \label{sec:experiments}


In the following case studies, we show how different hardware designs affect the level of averaging that can be done without loss of precision.
Our case studies contain a pipelined microprocessor running software AES, and an AES hardware accelerator. We chose AES-128 as its vulnerabilities are well-understood in the literature, and therefore this case makes a good driver for our experiments.
\vspace{-.3cm}
\subsection{Case Study 1: Software AES on a Pipelined Processor}
For our pipelined processor, we use
the five-stage \mbox{BRISC-V}\footnote{\url{https://ascslab.org/research/briscv/index.html}} implementation of RISC-V RV32I ISA (integrated into a system-on-chip \cite{cryptoeprint:2021:1236}).
The five stages in this core are instruction fetch, instruction decode and operand access, execution, memory access, and write back. 
We synthesize the design for the open standard cell library SkyWater 130nm and 50MHz frequency using Cadence Genus. 
In each gate-level simulation, we load the binary file of the compiled AES software code into the program memory of the processor, feed a new plaintext into the simulation, and collect the vcd file.
We then simulate the power traces of the SubBytes step in the first round using Cadence Joules. 
In the leakage assessment step, we run CPA on the simulated traces with a leakage model of the Hamming weight of the first round SBox output, \textit{i.e.}, \texttt{HW(SBox(k$\oplus$p))} where \texttt{k} is the key byte and \texttt{p} is the corresponding plaintext byte.

\autoref{lst:asm_skivav} shows the assembly code for SubBytes running on RISC-V. In this code, the SBox is implemented as a table look-up. Iterating 16 times by two nested loops, line 10 loads the SBox result for each state byte and line 11 stores the result in the memory location of the state byte. 
According to the tested leakage model (\texttt{HW(SBox(k$\oplus$p))}), both lines 10 and 11 will leak. The leakage will stem from any sequential or combinatorial part of the processor handling the SBox output. 

\vspace{-.4cm}
\begin{minipage}[t]{\linewidth}
\begin{lstlisting}[basicstyle=\scriptsize,numbers=left,xleftmargin=1cm,caption=Assembly code of AES SubBytes for RISC-V, label={lst:asm_skivav}]
250:  lui   a1,0x1
254:  addi  a1,a1,1032  # state address
258:  addi  a3,a1,4
25c:  addi  a1,a1,20
260:  lui   a2,0x1
264:  addi  a2,a2,212   # sbox address
268:  addi  a5,a3,-4    # start loop 1
26c:  lbu   a4,0(a5)    # start loop 2
270:  add   a4,a2,a4
274:  lbu   a4,0(a4)    # load sbox byte
278:  sb    a4,0(a5)    # store sbox byte
27c:  addi  a5,a5,1
280:  bne   a5,a3,26c   # end loop 2 
284:  addi  a3,a3,4
288:  bne   a3,a1,268   # end loop 1 
\end{lstlisting}
\end{minipage}

\autoref{fig:instr_flow_skivav} demonstrates the data dependency of instr. 11 on instr. 10 causing a stall in cycle 3 (data dependency between decode and execute stage in cycle 2) and another stall in cycle 4 (data dependency between decode and memory stage in cycle 3). Furthermore, there is a one-cycle delay for memory load instructions, which causes instr. 10 to stall in the memory stage cycle 4 and subsequently flush in the write-back stage. Once instr. 10 reaches the write-back stage, its result is forwarded to instr. 11 in the decode stage.  Consequently, this pipeline diagram demonstrates that cycles 3 through 7 cause leakage of the \texttt{HW(SBox(k$\oplus$p))} during the highlighted cells. Therefore, power samples taken from these cycles are constructive samples for our leakage model. Since we have a 5-cycle window of constructive samples and we don't control the starting point of averaged samples, we average 3 clock cycles to ensure only constructive samples are averaged. 

\begin{figure}[t]
    \centering
    \includegraphics[width=.8\linewidth]{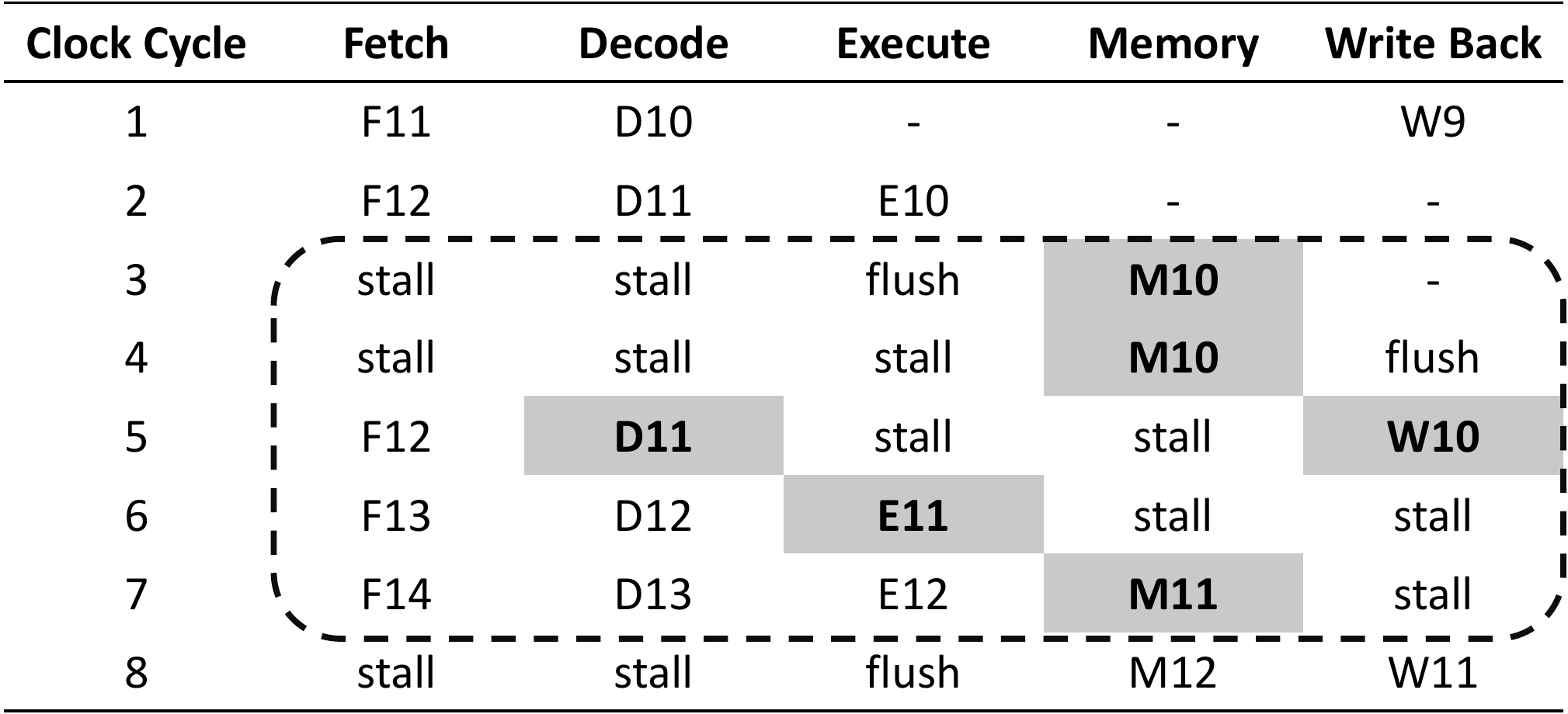}
    \vspace{-.3cm}
    \caption{Flow of instructions in \autoref{lst:asm_skivav} through five stages of RISC-V pipeline. Highlighted cells show the leaking parts.}
    \label{fig:instr_flow_skivav}
    \vspace{-.6cm}
\end{figure}

We simulate 1k power traces for RISC-V running the aforementioned assembly code and simulated traces with 1 sample per clock cycle (1s1cc) and 1 sample per 3 clock cycles (1s3cc). 
CPA on both trace sets resulted in $\rho > 0.999999$ for every key byte. The achieved high correlation despite the downsampled traces is thanks to the constructive samples.
Therefore, using the downsampled 1s1cc and 1s3cc traces, we can still find the leakage with highest possible correlation while reducing the power simulation time by 1.86$\times$.


Additionally, we simulate traces with even fewer samples (1 sample every 10, 20, 30, 40, and 50 cycles) to compare their correlation results in CPA. 
Although CPA is still able to reveal all key bytes successfully, 
the correlation coefficients for the heavily downsampled traces (1s10cc, 1s20cc, 1s30cc, 1s40cc, and 1s50cc) reduce. 
We can use heavily downsampled traces to quickly locate the leaky time periods and gradually increase the sample rate on the leaky regions to find more localized leaky points and therefore reduce the simulation time.
For instance, \autoref{fig:skivav-downsample-zoom} shows correlation coefficients from CPA on the last key byte of AES running on RISC-V.
\autoref{tab:skivav_speedup} shows the performance gain from increasing number of averaged samples for the same number of traces (normalized to 1s1cc).


\begin{figure}[h]
\vspace{-.3cm}
    \centering
    \includegraphics[width=.85\linewidth]{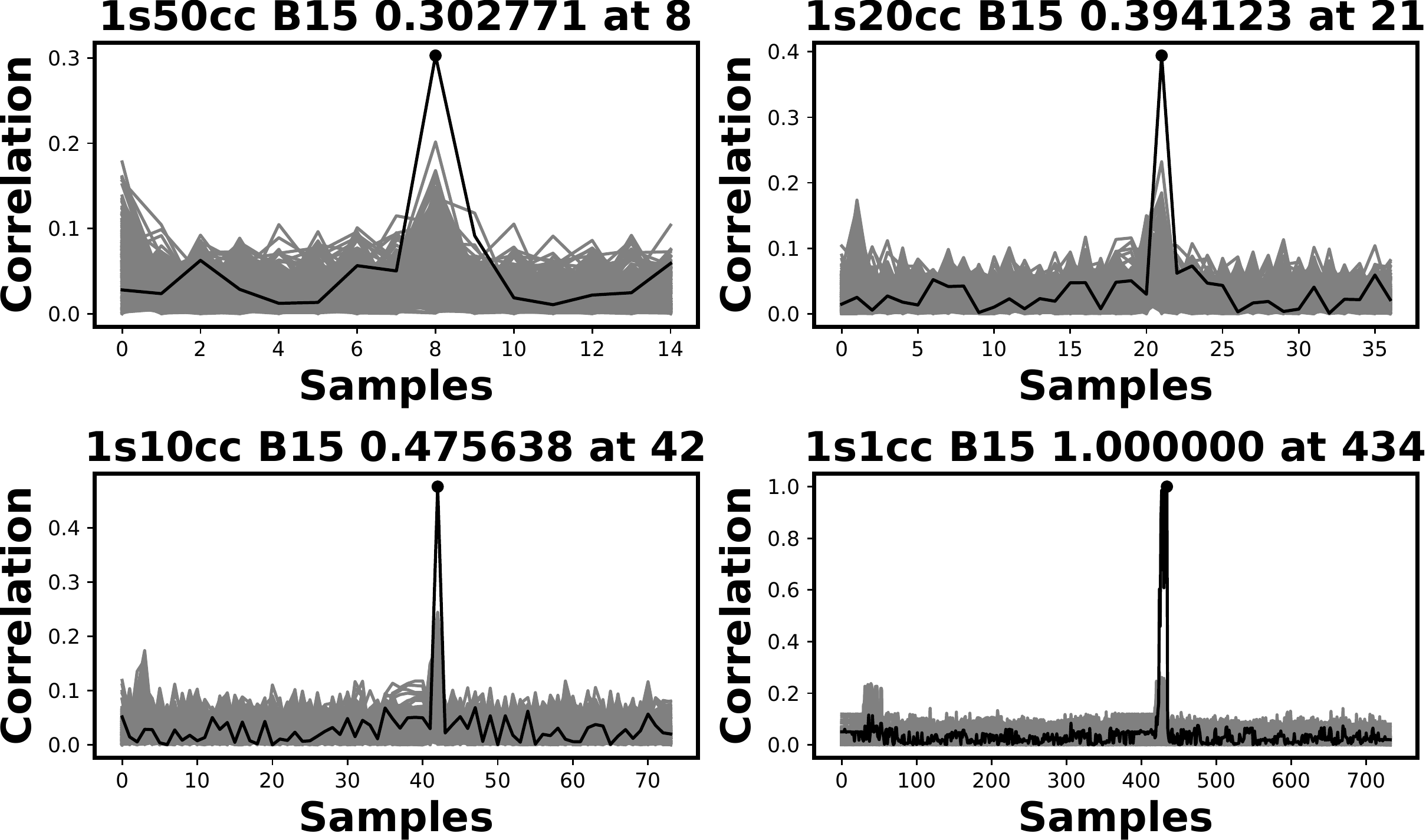}
    \vspace{-.4cm}
    \caption{CPA $\rho$ vs. sample number for locating the power leakage of last key byte (B15) from coarse (1s50cc) to finer (1s1cc) traces for software AES running on RISC-V. Grey lines are correlation values for incorrect key guesses and black line is the correlation value for correct key guess.}
    \label{fig:skivav-downsample-zoom}
    \vspace{-.5cm}
\end{figure}


\begin{table}[h]\centering
\vspace{-.4cm}
\caption{Speed-up of averaged RISC-V traces compared to 1s1cc}\label{tab:skivav_speedup}
\vspace{-.3cm}
\scriptsize
\begin{tabular}{lrrrrrr}\toprule
Num. of Averaged Clock Cycles &1s10cc &1s20cc &1s30cc &1s40cc &1s50cc \\\midrule
Simulation Speed-up &3.35$\times$ &4.86$\times$ &5.45$\times$ &5.78$\times$ &6.50$\times$ \\
\bottomrule
\end{tabular}
\vspace{-.5cm}
\end{table}


\vspace{-.2cm}
\subsection{Case Study 2: Hardware AES}
In the next experiment, we shift our focus from constructive samples to destructive samples. 
We consider a hardware accelerator for AES-128 that runs 4 SBoxes in parallel in each clock cycle, therefore, destructive samples are intrinsically present in each power sample.
\autoref{fig:picoaes} shows the structure of the AES implementation.
Each AES round is executed in 5 clock cycles. The first four clock cycles each execute 4 SBoxes (grey shade in \autoref{fig:picoaes}). The {\tt ma} (resp. {\tt mb}) multiplexers choose which 32bit part of the state (SW0 through SW3) should be updated (resp. calculated) through the SBox modules. 
The last clock cycle executes the rest of the blocks in succession ({\tt SR}, {\tt MC}, and AddRoundKey) following which {\tt ma} multiplexers choose the appropriate update values for each state word. 
{\tt mc} sets the correct input for the AddRoundKey step, \textit{i.e.}, plaintext ({\tt P}) at the start, {\tt MC} output for the first 9 rounds, and {\tt SR} output for the last round.
We synthesize the AES module for SkyWater 130nm standard cell library and 50MHz clock frequency using Cadence Genus and simulate 1k downsampled (1s1cc, 1s2cc, 1s3cc, 1s4cc) post-synthesis gate-level power traces using Cadence Joules. 

In the first round of AES, each {\tt SW} register will overwrite the first AddRoundKey output with its corresponding SBox output. Therefore, for CPA, we use \texttt{HD(k$\oplus$p,SBox(k$\oplus$p))} as the leakage model.
\autoref{fig:coaes_corr_zoom} shows the CPA correlation result for the last subkey as an example. The correlation values have drastically reduced and in some cases (e.g. 1s4cc for subkey 15) the CPA fails.
This example shows that a secure hardware designer should avoid combining destructive samples in the power analysis stage to prevent false negative conclusions in the assessments.

\begin{figure}[t]
    \centering
    \includegraphics[width=.75\linewidth]{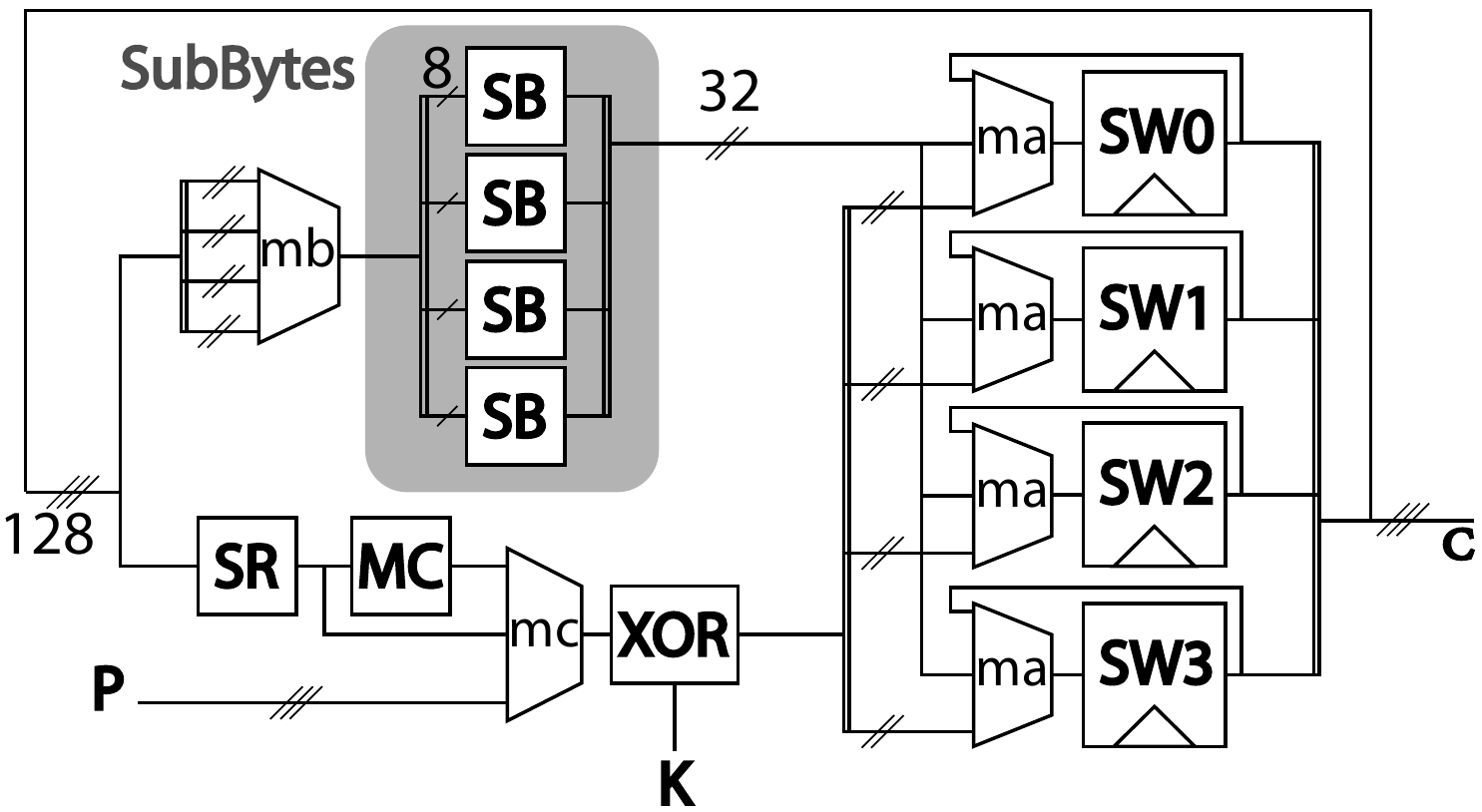}
    \vspace{-.4cm}
    \caption{Data path of the analyzed AES hardware accelerator. SB: SBox, SR: ShiftRows, MC: MixColumns.}
    \label{fig:picoaes}
    \vspace{-.6cm}
\end{figure}

\begin{figure}[h]
\vspace{-.3cm}
    \centering
    \includegraphics[width=.85\linewidth]{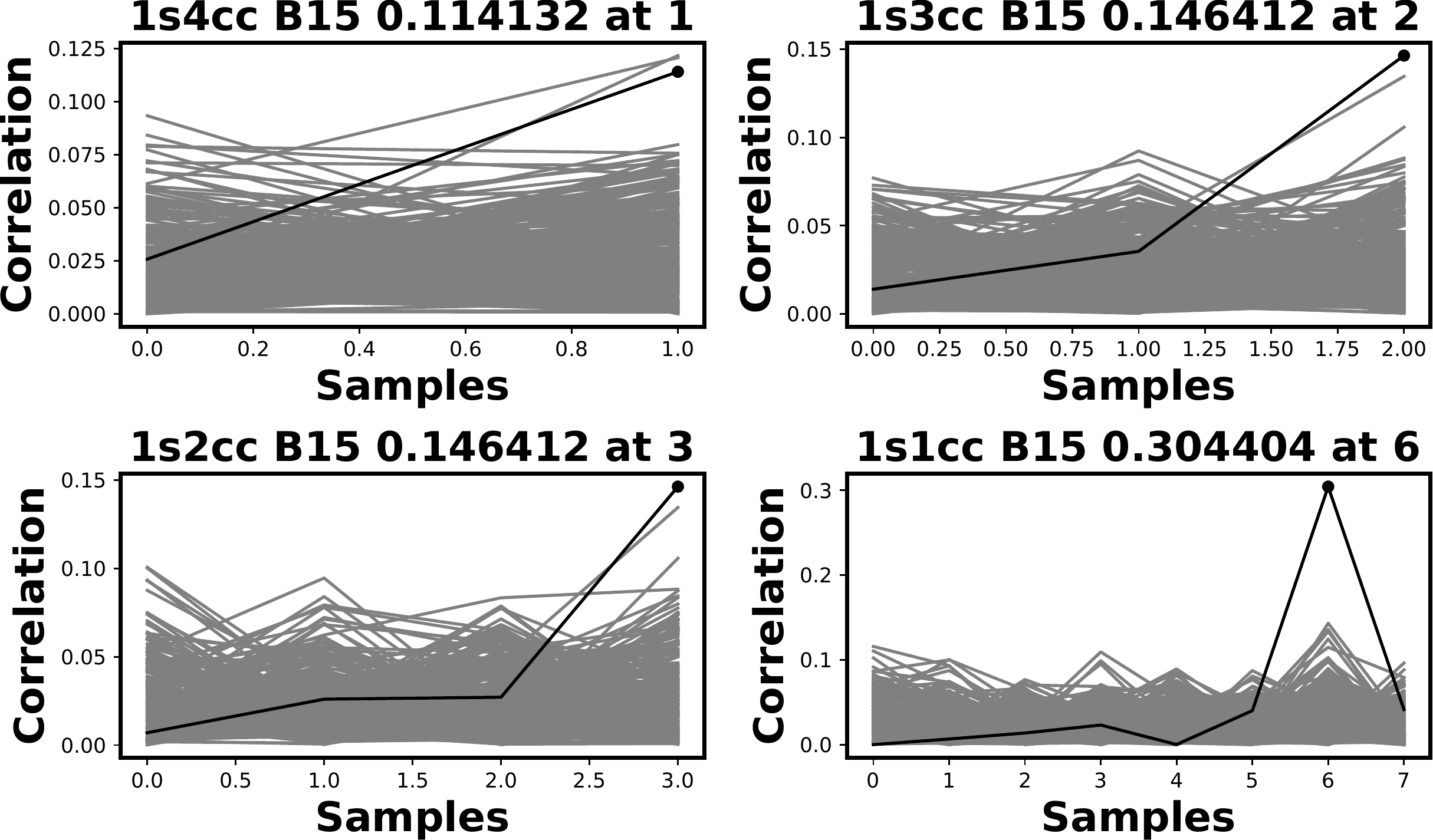}
    \vspace{-.4cm}
    \caption{CPA $\rho$ vs. sample number on last key byte (B15) for AES accelerator. Grey lines are correlation values for incorrect key guesses and black line is the correlation value for correct key guess.}
    \label{fig:coaes_corr_zoom}
\vspace{-.5cm}
\end{figure}


\textbf{Discussion.} As shown experimentally and theoretically, a secure hardware designer can take advantage of constructive samples in a design to downsample traces and reduce the power simulation time which is the bottleneck for pre-silicon PSCL assessment. At the same time, the designer should be aware that averaging destructive samples can lead to false negative conclusions in leakage assessment. 
Fortunately, classifying power samples as constructive or destructive is straightforward for well-known algorithms and implementations such as the ones analyzed in our case studies. However, additional preprocessing may be needed for unknown designs.
For instance, using specific test vectors that only exercise part of the key and applying statistical tests to tag a given power sample as destructive or constructive.
In our future work, we study techniques that can help a designer classify power samples as such for certain leakage criterion.


\vspace{-.3cm}
\section{Conclusion}\label{sec:conclusion}
We studied the implication of averaged sampling for pre-silicon side-channel leakage assessment.
We introduced the concept of constructive and destructive samples as categories of power samples which will facilitate or hinder the leakage assessment of a certain sensitive data. 
We showed how using downsampling in designs for which the constructive samples are well understood aid in reducing power simulation time without significant loss in precision.
In future work, we will study tests that can categorize samples as constructive or destructive to be applied to less-known designs. 
\vspace{-.3cm}
\begin{acks}
\vspace{-.1cm}
This research was supported in part by NSF award CNS-1931639.
\end{acks}

\vspace{-.2cm}

\bibliographystyle{ACM-Reference-Format}
\bibliography{Bibliography}
\vspace{-.3cm}


\end{document}